\documentclass{conm-p-l}
\usepackage{amsfonts}
\usepackage{amssymb}
\usepackage[dvips]{graphics}
\usepackage{euscript}
\usepackage{color}
\usepackage{bbm}
\usepackage[colorlinks]{hyperref}
\hypersetup{allcolors=blue}

\newtheorem{thm}{Theorem}[section]
\newtheorem{cor}[thm]{Corollary}
\newtheorem{lem}[thm]{Lemma}

\theoremstyle{definition}

\theoremstyle{remark}

\numberwithin{thm}{section}
\DeclareMathOperator{\RE}{Re}

\newcommand{\R}{{\mathord{\mathbb R}}}
\newcommand{\N}{{\mathord{\mathbb N}}}

\newcommand{\C}{{\mathord{\mathbb C}}}
\newcommand{\Z}{{\mathord{\mathbb Z}}}
\newcommand{\E}{{\mathord{\mathbb E}}}

\def\idty{{\mathchoice {\mathrm{1\mskip-4mu l}} {\mathrm{1\mskip-4mu l}} %
{\mathrm{1\mskip-4.5mu l}} {\mathrm{1\mskip-5mu l}}}}

\DeclareMathOperator{\Tr}{Tr}

\DeclareMathOperator{\supp}{supp}

\definecolor{applegreen}{rgb}{0.55, 0.71, 0.0}

\begin{document}

\title[]{Correlations in disordered quantum harmonic oscillator systems: The effects of excitations and quantum quenches}

\author[H. Abdul-Rahman]{Houssam Abdul-Rahman}
\address{Department of Mathematics\\
University of Arizona\\
Tucson, AZ 85721, USA}
\email{houssam@math.arizona.edu}
\author[R. Sims]{Robert Sims}
\address{Department of Mathematics\\
University of Arizona\\
Tucson, AZ 85721, USA}
\email{rsims@math.arizona.edu}
\author[G. Stolz]{G\"unter Stolz}
\address{Department of Mathematics\\
University of Alabama at Birmingham\\
Birmingham, AL 35294 USA}
\email{stolz@uab.edu}
\date{}

\begin{abstract}
We prove spatial decay estimates on disorder-averaged position-momentum correlations in
a gapless class of random oscillator models. First, we prove a decay estimate on
dynamic correlations for general eigenstates with a bound that depends on the
magnitude of the maximally excited mode. Then, we consider the situation of a quantum quench.
We prove that the full time-evolution of an initially chosen (uncorrelated) product state
has disorder-averaged correlations which decay exponentially in space, uniformly in time.

\end{abstract}

\maketitle

%%%%%
%
%  Intro . . .
%
%%%%%%%

\allowdisplaybreaks

\section{Introduction}

The mathematical investigation of disorder effects on quantum many-body systems, including, in particular, the phenomenon of many-body localization (MBL), is still in the early stages of its development. It has recently received strong attention in the physics and quantum information theory literature, see, e.g., \cite{NH,Vosketal,Agarwaletal} for recent reviews with extensive lists of references. Most mathematical results have been for models whose study can be fully reduced to the investigation of an effective one-particle Hamiltonian (i.e., without interaction) such as the Anderson model. Only few results go beyond such models.  This includes \cite{Imbrie}, which proposes a derivation of MBL for certain disordered quantum spin chains, based on an as yet unproven assumption on level statistics for these models. Also, \cite{Mastro} establishes exponential decay of ground state correlations for the Aubry-Andr\'e model (an XXZ chain in quasi-periodic field, which maps via the Jordan-Wigner transform to interacting Fermions). Recently, fully rigorous proofs of MBL in the droplet spectrum of the XXZ chain in random field (a regime extending beyond the ground state) have been given in \cite{BW} and \cite{EKS}. In particular, the latter establishes exponential clustering of all eigenstates throughout the droplet spectrum with respect to arbitrary local observables.

Models which can be fully reduced to an effective one-particle Hamiltonian include the XY spin chain in random transversal field, see the review \cite{ANSS}, and the Tonks-Girardeau gas subject to a random potential \cite{SW}. Here we present some new results on localization properties for another such model, disordered harmonic oscillator systems, as previously studied in \cite{NSS,NSS2, AR18}. In \cite{NSS} results on the many-body dynamics in the form of zero-velocity Lieb-Robinson bounds as well as exponential decay of dynamical correlations (exponential clustering) of the ground state and of thermal states of such systems were shown. \cite{NSS2} further investigated ground and thermal states by establishing an area law for their bipartite entanglement entropy. More recently in \cite{AR18}, 
area laws are proven for uniform ensembles of equal-excitation energy eigenstates where the surface area bound increases linearly in the total number of excitations (modes).

There is, of course, a long history of interesting results for deterministic oscillator models.
It is well-known, see for example \cite{CE, SCW} and references therein, that ground states of uniformly gapped oscillator models
satisfy exponential decay of correlations. Moreover, area laws for both ground and thermal states of gapped oscillator models
can be found e.g.\ in \cite{CEPD}, see also the well-referenced review \cite{ECP}. By contrast, as in \cite{NSS,NSS2, AR18}, we will consider models where
the gap above the ground state energy vanishes in the thermodynamic limit.
For the results we will prove, estimates like the above mentioned
deterministic results will not suffice, and we instead exploit disorder
effects.

Our first new result here, Theorem~\ref{thm:Corr:pq} below, shows that exponential clustering in disordered oscillator systems also holds for the dynamic position-momentum correlations of excited states. The bound obtained will only depend on the maximal local excitation number of these states, i.e., when expressed in terms of the corresponding free Boson system, for states with positive particle number density. This is desirable to show that the model is in the many-body localized phase, as the latter, if it exists for a given model, should cover an extended part of the energy spectrum of the system. We also point out that, as opposed to ground and thermal states, the excited states of oscillator systems are not quasi-free, a property used in most of the works on exactly solvable models mentioned above. While the creation of excited states out of the ground state of an oscillator system is a simple algebraic process, our result can still be seen as a simple example of the possibility to go beyond quasi-free states in the study of disordered many-body systems.

In our second result, Theorem~\ref{thm:quenched}, we study quenched position-momentum correlations of disordered oscillator systems. Quantum quenches and their effect on the non-equilibrium dynamics of quantum many-body systems have been frequently considered in physics, see, e.g., \cite{PAKSAM11} for a survey with many related references, as well as \cite{SPA14} for a discussion of quantum quenches in the context of many-body localization.

In its simplest form, a quantum quench refers to the investigation
of a quantum state which is initially prepared with respect to
one Hamiltonian and then time-evolved with respect to another. 
A common scenario is as follows. Consider a system defined
on a Hilbert space for which there is a natural 
bipartite decomposition into two subsystems, i.e. 
$\mathcal{H} = \mathcal{H}_1 \otimes \mathcal{H}_2$.
Denote by $H=H_1\otimes\idty+\idty\otimes H_2+I$ the Hamiltonian for the full system where
$I$ represents the interaction between the two subsystems. 
As an initial state take a product $\varphi= \varphi_1 \otimes \varphi_2$ which is uncorrelated with respect to the bipartite decomposition. 
The time-evolution $\varphi_t = e^{-iHt} \varphi$ of this initial state $\varphi$ under the {\it full} Hamiltonian dynamics may exhibit interesting behavior,
for example, non-trivial correlations may develop in time due to the interaction $I$.

In the disordered oscillator systems considered here we will assume that the two states $\varphi_1$ and $\varphi_2$ have exponentially clustered correlations with respect to the Hamiltonians $H_1$ and $H_2$ of the subsystem. We will show that the position-momentum correlations of the state will remain globally exponentially clustered, uniform in time and in the sizes of both subsystems. In fact, Theorem~\ref{thm:quenched} will be more general in allowing for the decomposition into an arbitrary number of subsystems.

Applications of Theorem~\ref{thm:quenched}, which we discuss in Section~\ref{sec:applications}, include the case where the initial product state consists of thermal states of the subsystems, so that the result of \cite{NSS} on exponential clustering of these states applies, or where one starts with products of eigenstates of the subsystems, so that our first result, Theorem~\ref{thm:Corr:pq}, can be applied in the subsystems. In physical terms these applications say that if each of the subsystems is localized in the sense of exponential decay of static correlations of eigenstates and thermal states within the subsystem, then this form of quantum quench yields no thermalization. In this context we include Theorem~\ref{thm:thermal:NSS}, as a technical result, proven in Appendix~\ref{sec:tempdep}, which improves results on thermal state correlations in \cite{NSS} by quantifying the temperature dependence.

\section{Model and Results}
\subsection{The Model} \label{sec:model}

For any dimension $d\geq 1$, we consider a coupled harmonic oscillator system, on a finite rectangular box $\Lambda:=([a_1,b_1] \times \ldots \times [a_d,b_d]) \cap\mathbb{Z}^d$, given by the Hamiltonian
\begin{equation}\label{eq:Ham}
H_{\Lambda}=\sum_{x\in\Lambda}\left( p_x^2+\frac{k_x}{2}q_x^2\right)+\sum_{{\tiny\begin{array}{c}
                                                                                \{x,y\}\subset\Lambda :\\
                                                                                |x-y|=1
                                                                              \end{array}}
}\lambda(q_x-q_y)^2.
\end{equation}
This Hamiltonian acts on the Hilbert space
\begin{equation}
\mathcal{H}_{\Lambda}=\bigotimes_{x\in\Lambda}L^2(\mathbb{R})=L^2(\mathbb{R}^{\Lambda})
\end{equation}
and $q_x$ and $p_x$ are, respectively, the position and momentum operators at the site $x \in \Lambda$. By standard results, these operators are self adjoint, on suitable domains, and satisfy the commutation relations
\begin{equation}
[q_x,q_y]=[p_x,p_y]=0, \text{ and } [q_x,p_y]=i\delta_{x,y}\idty \text{ for all } x,y\in\Lambda.
\end{equation}
Here $\delta_{x,y}$ is the Kronecker delta function.

For each $x\in\Lambda$, $k_x$ represents a variable spring constant. We introduce disorder by allowing the sequence $\{ k_x \}$ to be chosen as a sequence of
i.i.d.\ random variables. More precisely, we will assume that their common distribution $\mu$ is absolutely continuous,
\begin{equation} \label{eq:distr}
d\mu(k) = \nu(k)dk, \quad \mbox{with $\|\nu\|_{\infty} < \infty\;$ and $\;\supp \nu = [0,k_{max}]$}
\end{equation}
for some constant $k_{\max}<\infty$.

The Hamiltonian $H_{\Lambda}$ describes a convenient family of harmonic oscillators
that interact through nearest neighbor terms with a coupling strength of $\lambda>0$.
To be clear, the second sum in \eqref{eq:Ham} is taken over all undirected edges
$\{\{x,y\} \subset \Lambda :|x-y|=1\}$ where $|\cdot|$ denotes the 1-norm. With methods
similar to those of \cite{NSS}, the results we prove below generalize to a larger class of
disordered oscillator models; the caveat being that the assumptions on the effective one-particle
Hamiltonian, i.e. the analogue of (\ref{eq:eigenCorr}) below, would need to be verified on a case-by-case basis.
For ease of presentation, we restrict our attention to the model defined by (\ref{eq:Ham}) above.

As is well known, the analysis of general oscillator systems reduces, via Bogoliubov transformation,
to the analysis of an effective one-particle Hamiltonian, see, e.g., \cite{NSS} where this is reviewed for more general systems. In the specific case of (\ref{eq:Ham}), the
corresponding one-particle Hamiltonian is the finite volume Anderson model on $\ell^2(\Lambda)$, i.e.
\begin{equation}\label{eq:h0}
h_{\Lambda} = \lambda h_{0,\Lambda} + \frac{1}{2}k,
\end{equation}
where $h_{0,\Lambda}$ is the negative discrete Laplacian on $\Lambda$ and $\frac{1}{2}k$ an i.i.d.\ random potential. As we work in finite volume, the spectrum of $h_{\Lambda}$ is discrete and, under our assumption of an absolutely continuous distribution for the $k_x$, almost surely simple (as is seen easily by an analyticity argument).

By our assumptions, it is clear that $h_{\Lambda}$ is self-adjoint and moreover, using that $0\le h_{0,\Lambda} \le 4d$,
\begin{equation} \label{eq:spec}
\sigma(h_{\Lambda})\subseteq \left[\frac{1}{2}\min_{x\in\Lambda} k_x, \left(4d\lambda+\frac{k_{\max}}{2}\right)\right].
\end{equation}
This means that $h_{\Lambda}$ is invertible almost surely, but its inverse will not be uniformly bounded in the disorder. In fact, for large boxes $\Lambda$, the minimum of $\sigma(h_\Lambda)$ will be close to zero with high probability, due to the presence of large clusters where all $k_x$ are close to zero.

As a localization characteristic of $h_\Lambda$, we will assume that its {\it singular} eigenfunction correlators decay exponentially. More precisely, we assume that there exist constants $C<\infty$, $\eta>0$ and $0<s\le1$, independent of $\Lambda$, such that
\begin{equation}\label{eq:eigenCorr}
\mathbb{E}\left(\sup_{|g|\leq 1}|\langle\delta_x,h_{\Lambda}^{-\frac{1}{2}}g(h_{\Lambda})\delta_y\rangle|^s\right)\leq C e^{-\eta|x-y|},
\end{equation}
for all $x,y\in\Lambda$,
where $\mathbb{E}(\cdot)$ is the disorder average and $\{\delta_x\}_{x\in\Lambda}$ the canonical basis of $\ell^{2}(\Lambda)$. The supremum is taken over all functions $g:\mathbb{R}\rightarrow\mathbb{C}$ with $|g(x)|\leq 1$ and $g(h_{\Lambda})$ is defined by the functional calculus of symmetric matrices.

The non-standard feature of (\ref{eq:eigenCorr}) is the term $h_{\Lambda}^{-1/2}$. By the discussion above, this term does not have an a-priori norm bound and can thus not be absorbed into $g(h_{\Lambda})$. This term also is the reason for the inclusion of the moment $s$, which increases the applicability of (\ref{eq:eigenCorr}).

In the absence of the factor $h_{\Lambda}^{-1/2}$ the bound
\begin{equation}
\E\left(\sup_{|g|\le 1} |\langle\delta_x, g(h_{\Lambda})\delta_y\rangle| \right) \le Ce^{-\eta|x-y|}
\end{equation}
is a standard result for two regimes, e.g.\ \cite{AizenmanWarzel}: (i) in dimension $d=1$ (where $\Lambda$ is an interval) for any choice of the distribution density $\nu$, and (ii) in the large disorder regime for $d> 1$ (meaning that $\|\nu \|_{\infty}$ is sufficiently small, or $k_x$ is rescaled by a large parameter).

The singular eigenfunction correlators in (\ref{eq:eigenCorr}) were first introduced in \cite{NSS}. As is discussed there in Appendix A, (\ref{eq:eigenCorr}) holds for $d\ge 1$ and large disorder with $s=1$ (combing Propositions~A.1(b) and A.3(b)), and for $d=1$ and any $\nu$ with $s=1/2$ (Propositions A.1(c) and A.4(a)). Note that in the latter example, the validity of (\ref{eq:eigenCorr}) with $s=1/2$ does not trivially imply
that it is also valid with $s$ replaced by $1$, due to the fact that $|\langle\delta_x,h_{\Lambda}^{-\frac{1}{2}}g(h_{\Lambda})\delta_y\rangle|$ does not satisfy a uniform a-priori bound. Applications such as this are the reason we allow for the flexibility of $s$ in (\ref{eq:eigenCorr}).

As further discussed in Section~\ref{sec:Reduction}, the diagonalization of $h_{\Lambda}$ transforms $H_{\Lambda}$ into a model of free bosons
\begin{equation}\label{eq:H:freeBoson}
H_{\Lambda}=\sum_{k=1}^{|\Lambda|}\gamma_k(2B_k^* B_k+\idty),
\end{equation}
where $\gamma_k^2$ are the eigenvalues of $h_{\Lambda}$, and the operators $\{B_k\}$ satisfy canonical commutation relations (CCR) namely
\begin{equation}\label{eq:CCR}
[B_j,B_k]=[B_j^*,B_k^*]=0, \quad [B_j,B_k^*]=\delta_{j,k}\idty.
\end{equation}
In this case, there is a unique normalized vacuum state $\Omega \in \mathcal{H}_{\Lambda}$ corresponding to these $B$-operators, i.e., $\Omega$ satisfies $B_k \Omega=0$ for all $k$.
An explicit orthonormal basis (ONB) of eigenvectors of $H_{\Lambda}$ is then given by
\begin{equation} \label{def:psi}
\psi_\alpha=\prod_{j=0}^{|\Lambda|}\frac{1}{\sqrt{\alpha_j!}}(B_j^*)^{\alpha_j}\Omega
\end{equation}
for an excitation number configuration $\alpha=(\alpha_1,\ldots,\alpha_{|\Lambda|}) \in \mathbb{N}_0^{|\Lambda|}$ (here $\mathbb{N}_0:= \{ 0, 1, 2, \cdots \}$). One easily checks that these excited states satisfy
\begin{equation} \label{eq:energies}
H_{\Lambda} \psi_{\alpha} = E_{\alpha} \psi_{\alpha} \quad \mbox{with} \quad E_\alpha=\sum_{k=1}^{|\Lambda|}(2\alpha_k+1)\gamma_k
\end{equation}
and therefore, the gap above the ground state energy of $H_{\Lambda}$ is $2\min_k\gamma_k$.

\subsection{Dynamic Correlations in Eigenstates} \label{sec:dyn_cor}

One goal of this work is to estimate dynamic correlations of
position and momentum operators in arbitrary eigenstates. To make this more precise, let $\tau_t(A)$ denote the Heisenberg evolution of an operator $A$ under $H_\Lambda$, i.e.,
\begin{equation} \label{Heisenberg}
\tau_t(A)=e^{itH_{\Lambda}}Ae^{-itH_{\Lambda}},
\end{equation}
and for any trace-class operator $\rho$ on $\mathcal{H}_\Lambda$ take
\begin{equation} \label{Schrodinger}
\rho_t = e^{-itH_{\Lambda}} \rho e^{itH_\Lambda}
\end{equation}
to be the Schr\"odinger evolution of $\rho$. In this case, if $\langle A\rangle_{\rho}=\Tr[A\rho]$ denotes the $\rho$-expectation of the observable $A$,
then the Heisenberg and Schr\"odinger evolutions are related by $\langle \tau_t(A) \rangle_{\rho} = \langle A \rangle_{\rho_t}$.

It is convenient to introduce a $2|\Lambda|\times 2|\Lambda|$ correlation matrix
\begin{equation} \label{mixedcor}
\Gamma_{\rho}(t,t') : = \left\langle \tau_t \begin{pmatrix} q \\ p \end{pmatrix} (q^T, p^T) \right\rangle_{\rho_{t'}} -
\left\langle \tau_t \begin{pmatrix} q \\ p \end{pmatrix} \right\rangle_{\rho_{t'}} \left\langle (q^T, p^T) \right\rangle_{\rho_{t'}}
\end{equation}
which collects {\it mixed-time} dynamic correlations of position and momentum operators corresponding to $\rho$. Here $\begin{pmatrix} q\\ p \end{pmatrix}$ and $(q^T,p^T)$ are $2|\Lambda|$ column and row vectors, the time-evolution and $\rho$-expectation of vectors and matrices are understood component-wise, and columns are multiplied with rows in the usual sense of matrix multiplication to form matrices. Use of these mixed-time correlations $\Gamma_{\rho}(t,t')$ is convenient when formulating our main results below; in one we set $t'=0$ and in the other $t=0$.

Our first result concerns disorder-averaged correlations in eigenstates, i.e. we consider
$\rho= \rho_{\alpha} = | \psi_{\alpha} \rangle \langle \psi_{\alpha}|$ for some $\alpha \in \mathbb{N}_0^{| \Lambda|}$. Since eigenstates are time-invariant,
we set $t'=0$ in (\ref{mixedcor}) and denote by $\Gamma_{\alpha}(t) := \Gamma_{\rho_{\alpha}}(t,0)$. Note that
\begin{equation} \label{eq:exp:qp}
\langle \tau_t(q_x)\rangle_{\rho_{\alpha}} = \langle \tau_t(p_x)\rangle_{\rho_{\alpha}} = 0
\end{equation}
for all $x$, $t$ and $\alpha$ (in fact, we will also see this directly in the proof of Theorem~\ref{thm:Corr:pq} below). Thus (\ref{mixedcor}) simplifies to
\begin{equation} \label{dyn_eigen_cor}
\Gamma_{\alpha}(t) = \left\langle \tau_t \begin{pmatrix} q \\ p \end{pmatrix} (q^T, p^T) \right\rangle_{\rho_{\alpha}}.
\end{equation}

Lastly, for a $2|\Lambda| \times 2|\Lambda|$ block-matrix
\begin{equation} \label{eq:block}
M= \begin{pmatrix} A & B \\ C & D \end{pmatrix}, \; \mbox{let} \; M_{xy} = \begin{pmatrix} A_{xy} & B_{xy} \\ C_{xy} & D_{xy} \end{pmatrix}
\end{equation}
 be $2\times 2$-sub-matrices with matrix norms $\|M_{xy}\|$. For definiteness we choose the latter to be the Euclidean operator norm.

\begin{thm}\label{thm:Corr:pq}
Assume that the effective Hamiltonian $h_{\Lambda}$ satisfies (\ref{eq:eigenCorr}), with bounds uniform in $\Lambda$. Then
\begin{equation}\label{eq:corr:qp}
\mathbb{E}\left(\sup_{t} \|(\Gamma_{\alpha}(t))_{xy}\|^s \right) \leq C C' (1+\|\alpha\|_\infty)^{1+s}e^{-\eta|x-y|}
\end{equation}
for all finite rectangular boxes $\Lambda \subset \Z^d$, $x,y\in\Lambda$ and $\alpha\in\N_0^{| \Lambda|}$. Here $C$, $\eta$ and $s$ are as in (\ref{eq:eigenCorr}) and $C'<\infty$ depends on $d$, $\lambda$, $s$ and $k_{max}$, but is independent of $\Lambda$.
\end{thm}

We finally note that our proof of Theorem~\ref{thm:Corr:pq} in Section~\ref{sec:Proof:cluster} below will show, the bound in (\ref{eq:corr:qp}) can be slightly improved to
\begin{equation}\label{eq:corr:qp:2}
\leq \tilde{C} \left(1+\min\{\|\alpha\|_1, \|\alpha\|_\infty^{1+s}\}\right)e^{-\eta|x-y|},
\end{equation}
which is better for excitation vectors with only a few large excitations $\alpha_j$ (say, just one of them).

\subsection{Quenched Correlations}\label{sec:Results:Quenshed}
Our second result concerns the position and momentum correlations when a quantum quench is applied. In particular, we  decompose the rectangular box $\Lambda$ into $M$ disjoint rectangular sub-boxes
\begin{equation} \Lambda =\biguplus_{\ell=1}^M\Lambda_\ell.
\end{equation}
For $\ell=1,2,\ldots,M$, consider the restrictions $H_{\Lambda_\ell}$ of the harmonic system $H_{\Lambda}$ to $\Lambda_\ell$.
Let $H_{0,\Lambda}$ denote the Hamiltonian of the non-interacting system on $\mathcal{H}_{\Lambda}$,
\begin{equation}
H_{0,\Lambda}=\sum_{\ell=1}^M H_{\Lambda_\ell} \otimes \idty_{\Lambda \setminus\Lambda_\ell}.
\end{equation}
For each $\ell=1,\ldots,M$, let $\rho_\ell$ be a state acting on the Hilbert space $\mathcal{H}_{\Lambda_\ell}:=L^2(\mathbb{R}^{\Lambda_\ell})$. In particular, we will consider the cases where the $\rho_\ell$ are either eigenstates or thermal states of $H_{\Lambda_\ell}$.
We are interested in the Schr\"odinger time evolution $\rho_t$, under the full Hamiltonian $H_{\Lambda}$ given in (\ref{eq:Ham}), of the state initially given by the product state
\begin{equation}\label{eq:product}
\rho:= \bigotimes_{\ell=1}^M\rho_\ell.
\end{equation}
This {\it quantum quench} is understood as a sudden change in the Hamiltonian $H_{0,\Lambda}$ at $t=0$, consisting in switching on the interactions between the subsystems $H_{\Lambda_\ell}$.

To describe the dynamic correlations in this case we set the first argument equal to zero in (\ref{mixedcor}) and define
\begin{equation} \label{eq:quenchcor}
\tilde{\Gamma}_{\rho}(t) := \Gamma_{\rho}(0,t).
\end{equation}
That the local systems $H_{\Lambda_\ell}$ are initially uncorrelated means that for $x\in\Lambda_j$ and $y\in\Lambda_\ell$ with $j\neq \ell$,
\begin{equation}
(\tilde{\Gamma}_{\rho}(0))_{xy}=0.
\end{equation}

The following result says that if each of the subsystems $H_{\Lambda_\ell}$ is localized in the sense of exponential decay of static correlations within the subsystem, then their quantum quench, described above, yields no thermalization. More precisely, the so-called quenched dynamic correlations of the product state $\rho$ remain exponentially decaying in the fully interacting system for all times (here we also use the local $2|\Lambda_\ell| \times 2|\Lambda_\ell|$ correlation matrices $\Gamma_{\rho_\ell}$ with $t=t'=0$ in (\ref{mixedcor})).

\begin{thm}\label{thm:quenched}
Assume that the effective Hamiltonian $h_{\Lambda}$ satisfies (\ref{eq:eigenCorr}), with bounds uniform in $\Lambda$. Let $\rho_\ell \in\mathcal{B}(\mathcal{H}_{\Lambda_\ell})$, $\ell=1,\ldots,M$, be a family of states such that, for some $C'<\infty$, and $\eta'>0$,
\begin{equation}\label{AssumLocal}
\mathbb{E}\left(\|(\Gamma_{\rho_\ell})_{xy}\|^s\right)\leq C' e^{-\eta' |x-y|}
\end{equation}
for all $\ell$ and all $x,y\in\Lambda_\ell$, where $0<s\le 1$ is as in (\ref{eq:eigenCorr}).

Then, for $\eta$ from (\ref{eq:eigenCorr}), $\tilde{\eta}:= \frac{1}{6}\min\{\eta, \eta'\}$ and $\rho = \bigotimes_\ell \rho_\ell$, there exists a constant $C''<\infty$ such that
\begin{equation}\label{eq:quenched}
\mathbb{E}\left(\sup_{t\in \R} \|(\tilde{\Gamma}_{\rho}(t))_{xy}\|^\frac{s}{3}
\right)\leq (C')^{1/3}C'' e^{-\tilde{\eta}|x-y|}
\end{equation}
for all $x,y\in\Lambda$. Here $C'$ is the constant from (\ref{AssumLocal}) and $C''$ depends on $d$, $\lambda$, $s$, $k_{max}$ and $\tilde{\eta}$, but is independent of $\Lambda$ and the number of subregions $M$.
\end{thm}

This will be proven in Section~\ref{sec:Proof:Quenched}.

By results in \cite{NSS}, special cases where condition (\ref{AssumLocal}) is known to hold include the ground state and thermal states of the subsystems $H_{\Lambda_\ell}$. Theorem~\ref{thm:Corr:pq} above extends this to excited states. In each of these cases, (\ref{AssumLocal}) actually follows from (\ref{eq:eigenCorr}). Theorem~\ref{thm:quenched} allows for the additional freedom to choose different temperatures and different maximal excitation numbers in each of the subsystems, or even to choose some of the factors in the initial product state as thermal states and others as excited states. It is then of some interest to understand the dependence of the constants in (\ref{eq:quenched}) on these additional parameters. We include a more thorough discussion of this in Section~\ref{sec:applications} at the end of this paper.

%%%
%
%  Notation etc.
%
%
%%%%%%%%%%%%

\section{Reduction to the Effective Hamiltonian}\label{sec:Reduction}

In this section, we briefly review the previously mentioned reduction of the many-body Hamiltonian $H_{\Lambda}$
to the effective one-particle $h_{\Lambda}$ as a means to introduce some relevant notation. Once this is done, we provide a
simple lemma concerning mixed-time correlations, i.e. (\ref{mixedcor}).

Keeping with the vector notation established in Section~\ref{sec:dyn_cor}, one readily sees that
the oscillator Hamiltonian $H_{\Lambda}$ in (\ref{eq:Ham}) can be re-written as
\begin{equation} \label{base_ham}
H_{\Lambda} = (q^T, p^T) \begin{pmatrix} h_{\Lambda} & 0 \\ 0 & \idty \end{pmatrix} \begin{pmatrix} q \\ p \end{pmatrix}
\end{equation}
with $h_{\Lambda}$ the effective one-particle Hamiltonian described in (\ref{eq:h0}).
The real non-negative matrix $h_{\Lambda}$ can be diagonalized in terms of a real orthogonal $\mathcal{O}: \C^{|\Lambda|} \to \ell^2(\Lambda)$ and its transpose $\mathcal{O}^T = \mathcal{O}^{-1}$, i.e.
\begin{equation}
\mathcal{O}^T h_{\Lambda} \mathcal{O} = \gamma^2
\end{equation}
where $\gamma^2 = {\rm diag}( \gamma_k^2)$ with $1 \leq k \leq | \Lambda|$. Here the
numbers $\gamma_k^2$ are the eigenvalues of $h_{\Lambda}$ counted according to multiplicity.
By our assumptions on the spring constants, the eigenvalues of $h_{\Lambda}$ are almost surely positive, and we will denote by
$\gamma = {\rm diag}( \gamma_k)$ with $\gamma_k>0$ for all $1 \leq k \leq | \Lambda|$. As discussed in Section~\ref{sec:model}, the $\gamma_k$ are almost surely non-degenerate.

As is well-known, see \cite{NSS} for more details in this specific setting, $H_{\Lambda}$ can be reduced to a system of free Bosons. In fact, consider the mapping
\begin{equation} \label{def:Vinv}
V^{-1} = \frac{1}{\sqrt{2}} \begin{pmatrix} \idty & i \idty \\ \idty & - i \idty \end{pmatrix} \begin{pmatrix} \gamma^{1/2} \mathcal{O}^T & 0 \\ 0 & \gamma^{-1/2} \mathcal{O}^T \end{pmatrix}.
\end{equation}
Our assumptions guarantee this map is almost surely well-defined, invertible, and one readily checks that the product
\begin{equation} \label{def:B}
\begin{pmatrix} B \\ B^* \end{pmatrix} : = V^{-1} \begin{pmatrix} q \\ p \end{pmatrix}
\end{equation}
produces a collection of operators $\{B_k \}_{k=1}^{| \Lambda|}$ on $\mathcal{H}_{\Lambda}$ which, together with their adjoints, satisfy the CCR, i.e. (\ref{eq:CCR}). Moreover,
in terms of these $B$-operators
\begin{equation} \label{free_Bos}
H_\Lambda = \sum_{k=1}^{| \Lambda|} \gamma_k ( 2 B_k^*B_k + \idty),
\end{equation}
a model of free Bosons.

Due to the simple form of (\ref{free_Bos}), the dynamics of these $B$-operators is
\begin{equation}
\tau_t \begin{pmatrix} B \\ B^* \end{pmatrix} = \begin{pmatrix} e^{-2it \gamma} & 0 \\ 0 & e^{2it \gamma}  \end{pmatrix} \begin{pmatrix} B \\ B^* \end{pmatrix}
\end{equation}
{f}rom which the dynamics of position and momentum operators readily follows,
\begin{equation} \label{dyn_qp}
\tau_t \begin{pmatrix} q \\ p \end{pmatrix} = V \begin{pmatrix} e^{-2it \gamma} & 0 \\ 0 & e^{2it \gamma}  \end{pmatrix} \begin{pmatrix} B \\ B^* \end{pmatrix},
\end{equation}
where we have used (\ref{def:B}). It will also be convenient to note that
\begin{equation} \label{def:V}
 V =  \frac{1}{\sqrt{2}}  \begin{pmatrix}  \mathcal{O} \gamma^{-1/2} & 0 \\ 0 & \mathcal{O}  \gamma^{1/2} \end{pmatrix} \begin{pmatrix} \idty &  \idty \\ -i \idty & i \idty \end{pmatrix}.
\end{equation}

As indicated in Section~\ref{sec:dyn_cor}, much of our analysis reduces to the investigation of
the mixed-time correlation function $\Gamma_{\rho}(t,t')$ in (\ref{mixedcor}) for a state $\rho$.
 Since (\ref{dyn_qp}) shows that the dynamics of position and momentum operators
can be expressed in terms of the $B$-operators, up to scalar-valued coefficients, one immediately has the following.
\begin{lem} \label{lem:Cor_calc}
Let $\rho$ be a state on $\mathcal{H}_\Lambda$. Suppose that the matrix
\begin{equation} \label{Cor_B}
\Gamma_{\rho}^B = \left\langle \begin{pmatrix} B \\ B^* \end{pmatrix} (B^T, (B^*)^T) \right\rangle_{\rho} - \left\langle \begin{pmatrix} B \\ B^* \end{pmatrix} \right\rangle_{\rho} \left\langle (B^T, (B^*)^T) \right\rangle_{\rho}
\end{equation}
is well-defined. Then, for all $t,t' \in \mathbb{R}$,
\begin{equation} \label{mixedtimecor}
\Gamma_{\rho}(t,t') = V \begin{pmatrix}  e^{-2i(t+t') \gamma} & 0 \\ 0 & e^{2i(t+t') \gamma} \end{pmatrix} \Gamma_{\rho}^B \begin{pmatrix}  e^{-2it' \gamma} & 0 \\ 0 & e^{2it' \gamma}  \end{pmatrix} V^T
\end{equation}
where $V$ is as in (\ref{def:V}).
\end{lem}

%%%%%%%%
%
% Eigenstate Correlations
%
%
%%%%%%%%%%%%

\section{Proof of Theorem~\ref{thm:Corr:pq}} \label{sec:Proof:cluster}

 We begin with a calculation which evaluates the eigenstate correlation matrix $\Gamma_{\alpha}(t)$ given by (\ref{dyn_eigen_cor}) in terms of the effective Hamiltonian $h_\Lambda$, using Lemma~\ref{lem:Cor_calc}.

\begin{lem}\label{lem:Cor_Mat}
We have the identity
\begin{eqnarray} \label{gen_cor_alpha}
\Gamma_{\alpha}(t) & = & \begin{pmatrix} \mathcal{O} \alpha \mathcal{O}^T & 0 \\ 0 & \mathcal{O} \alpha \mathcal{O}^T \end{pmatrix} \begin{pmatrix} h_\Lambda^{-1/2} \cos(2 t h_\Lambda^{1/2}) & \sin(2 t h_\Lambda^{1/2}) \\ -\sin(2 t h_\Lambda^{1/2}) & h_\Lambda^{1/2} \cos(2t h_\Lambda^{1/2}) \end{pmatrix} + \\ \nonumber
& \mbox{ } & \quad + \frac{1}{2} \begin{pmatrix} h_\Lambda^{-1/2} e^{-2i t h_\Lambda^{1/2}} & ie^{-2i t h_\Lambda^{1/2}} \\ -ie^{-2i t h_\Lambda^{1/2}} & h_\Lambda^{1/2} e^{-2it h_\Lambda^{1/2}} \end{pmatrix}.
\end{eqnarray}
\end{lem}

Here, in a slight abuse of notation, we use $\alpha$ also to denote the diagonal matrix with entries $\alpha_k$, $1 \leq k \leq |\Lambda|$.

\begin{proof}
By orthogonality of the eigenvectors $\psi_{\alpha}$, it is clear that each of
\begin{equation}
\left\langle \begin{pmatrix} B \\ B^* \end{pmatrix} \right\rangle_{\rho_{\alpha}} , \langle BB^T \rangle_{\rho_\alpha} , \mbox{ and } \langle B^*(B^*)^T \rangle_{\rho_\alpha}
\end{equation}
vanish identically. We note that this and (\ref{dyn_qp}) implies (\ref{eq:exp:qp}). Moreover, using also the commutation relations (\ref{eq:CCR}), we find that for all $1 \leq j,k \leq | \Lambda|$,
\begin{equation}
\langle B^*_k B_j \rangle_{\rho_\alpha} + \delta_{j,k} = \langle B_j B^*_k \rangle_{\rho_\alpha} = \langle B^*_j \psi_{\alpha},  B^*_k \psi_{\alpha} \rangle = ( \alpha_j + 1) \delta_{j,k} \, .
\end{equation}
and therefore, for $\Gamma_{\rho_\alpha}^B$ as in (\ref{Cor_B}), we have that
\begin{equation}
\Gamma_{\rho_\alpha}^B  = \begin{pmatrix} 0 & ( \alpha + \idty) \\ \alpha & 0 \end{pmatrix}
\end{equation}
An application of Lemma~\ref{lem:Cor_calc} yields
\begin{equation}
\Gamma_{\alpha}(t) = V \begin{pmatrix} 0 & (\alpha+\idty) e^{-2i \gamma t} \\  \alpha e^{2 i \gamma t} & 0 \\ \end{pmatrix} V^T.
\end{equation}

A short calculation shows that
\begin{equation}
\frac{1}{2} \begin{pmatrix} \idty & \idty \\ -i \idty & i \idty \end{pmatrix} \begin{pmatrix} 0 & (\alpha+\idty) e^{-2i \gamma t} \\  \alpha e^{2 i \gamma t} & 0 \\ \end{pmatrix} \begin{pmatrix} \idty & -i \idty \\ \idty & i \idty \end{pmatrix}
\end{equation}
can be rewritten as
\begin{equation}
\begin{pmatrix} \alpha & 0 \\ 0 & \alpha \end{pmatrix} \begin{pmatrix} \cos(2 \gamma t) & \sin(2 \gamma t) \\ - \sin(2 \gamma t) & \cos(2 \gamma t) \end{pmatrix}
+ \frac{1}{2} \begin{pmatrix} e^{-2i \gamma t} & 0 \\ 0 & e^{-2i \gamma t} \end{pmatrix} \begin{pmatrix} \idty & i \idty \\ - i \idty & \idty \end{pmatrix}.
\end{equation}
Using the form of $V$ and $V^{-1}$ in (\ref{def:V}) and (\ref{def:Vinv}) this gives
\begin{eqnarray}
\Gamma_{\alpha}(t) & = & \begin{pmatrix} \mathcal{O} \gamma^{-1/2} & 0 \\ 0 &  \mathcal{O} \gamma^{1/2} \end{pmatrix}
\begin{pmatrix} \alpha \cos(2 \gamma t) & \alpha \sin(2 \gamma t) \\ - \alpha \sin(2 \gamma t) & \alpha \cos(2 \gamma t) \end{pmatrix}
\begin{pmatrix}  \gamma^{-1/2} \mathcal{O}^T & 0 \\ 0 &  \gamma^{1/2}  \mathcal{O}^T \end{pmatrix}  \nonumber \\
& & \mbox{} + \frac{1}{2} \begin{pmatrix} \mathcal{O} \gamma^{-1/2} & 0 \\ 0 &  \mathcal{O} \gamma^{1/2} \end{pmatrix}
\begin{pmatrix}  e^{-2i \gamma t} & i  e^{-2i \gamma t} \\ - i  e^{-2i \gamma t}  &  e^{-2i \gamma t} \end{pmatrix}
\begin{pmatrix}  \gamma^{-1/2} \mathcal{O}^T & 0 \\ 0 &  \gamma^{1/2}  \mathcal{O}^T \end{pmatrix}.
\end{eqnarray}
This is the same as (\ref{gen_cor_alpha}), due to
\begin{equation}
\mathcal{O} \gamma^{-1/2} \alpha \cos(2 \gamma t) \gamma^{-1/2} \mathcal{O}^T = \mathcal{O} \alpha \mathcal{O}^T  h_\Lambda^{-1/2} \cos(2 t h_\Lambda^{-1/2})
\end{equation}
and similar consequences of the functional calculus.
\end{proof}

We can now present the proof of Theorem \ref{thm:Corr:pq}.
\begin{proof}(of Theorem \ref{thm:Corr:pq})
We first consider the most singular case, and then comment on how the remaining cases follow similarly.

Lemma~\ref{lem:Cor_Mat} demonstrates that
\begin{equation} \label{qq_cor}
\langle \tau_t(q_x) q_y \rangle_{\alpha} = \langle \delta_x, \mathcal{O} \alpha \mathcal{O}^T h_\Lambda^{-1/2} \cos( 2 t h_\Lambda^{1/2}) \delta_y \rangle + \frac{1}{2} \langle \delta_x, h_\Lambda^{-1/2} e^{- 2it h_\Lambda^{1/2}} \delta_y \rangle.
\end{equation}
Using (\ref{eq:eigenCorr}), it is clear that
\begin{equation}
\mathbb{E} \left( \sup_{t \in \mathbb{R}} \left| \langle \delta_x, h_\Lambda^{-1/2} e^{- 2it h_\Lambda^{1/2}} \delta_y \rangle \right|^s \right) \leq Ce^{- \eta|x-y|},
\end{equation}
and so we need only estimate the first term in (\ref{qq_cor}) above. If the eigenvalues of $h_{\Lambda}$ are non-degenerate, which holds almost surely, we can write
\begin{equation} \label{eq:efexpansion}
\mathcal{O}\alpha\mathcal{O}^T=\sum_{j=1}^{|\Lambda|}\alpha_j\chi_{\{\gamma_j^2\}}(h_\Lambda)=\sum_{a=0}^{\|\alpha\|_\infty}a \chi_{J(a)}(h_\Lambda).
\end{equation}
Here $\chi_{\{\gamma_j^2\}}(h_\Lambda)$ is the projection onto the eigenvector of $h_\Lambda$ to $\gamma_j^2$, and $\chi_{J(a)}(h_\Lambda)$ is the spectral projection for $h_\Lambda$ onto $J(a):=\{\gamma_j^2: \alpha_j =a\}$.
Given this, one immediately sees that
\begin{equation}
\left| \langle \delta_x, \mathcal{O} \alpha \mathcal{O}^T h_\Lambda^{-1/2} \cos( 2 t h_\Lambda^{1/2}) \delta_y \rangle \right|^s \leq
\sum_{a=0}^{\| \alpha \|_{\infty}} a^s \left| \langle \delta_x, h_\Lambda^{-1/2} \chi_{J(a)}(h_\Lambda) \cos( 2 t h_\Lambda^{1/2}) \delta_y \rangle \right|^s
\end{equation}
and therefore an application of (\ref{eq:eigenCorr}) again implies
\begin{equation}
\mathbb{E} \left( \sup_{t \in \mathbb{R}} \left| \langle \delta_x, \mathcal{O} \alpha \mathcal{O}^T h_\Lambda^{-1/2} \cos( 2 t h_\Lambda^{1/2}) \delta_y \rangle \right|^s \right)
\leq C \| \alpha \|_{\infty}^s ( \| \alpha \|_{\infty} +1) e^{- \eta |x-y|}.
\end{equation}
This completes the argument for the most singular correlations.

As is clear from Lemma~\ref{lem:Cor_Mat}, the other correlations in the $2\times2$-matrix $(\Gamma_\alpha(t))_{xy}$ produce similar terms. These terms require bounds on eigenfunction correlators less singular than (\ref{eq:eigenCorr}), in the sense that the term $h_{\Lambda}^{-1/2}$ is replaced by $\idty$ or $h_{\Lambda}^{1/2}$. They can be bounded by (\ref{eq:eigenCorr}) due to the uniform spectral bound (\ref{eq:spec}). For example, associating $\tilde{g}(x) := x^{1/2} g(x)$ with each $g$ such that $|g|\le 1$, one gets
\begin{equation}
\E \left(\sup_{|g|\le 1} |\langle \delta_x, g(h_{\Lambda}) \delta_y \rangle|^s \right) \le (4d\lambda+\frac{k_{max}}{2})^{s/2} \E \left( \sup_{|g|\le 1} |\langle \delta_x, h_{\Lambda}^{-1/2} g(h_\Lambda) \delta_y \rangle|^s \right),
\end{equation}
and similar for $\E \left(\sup_{|g|\le 1} |\langle \delta_x, h_{\Lambda}^{1/2} g(h_{\Lambda}) \delta_y \rangle|^s \right)$.

Finally, the bound in terms of $\|\alpha\|_1$ in (\ref{eq:corr:qp:2}) follows by directly considering the middle term in (\ref{eq:efexpansion}).

This completes the proof of Theorem~\ref{thm:Corr:pq}.

\end{proof}

%%%%%%%
%
%  On Quenching . . .
%
%%%%%%%%%%%

\section{Proof of Theorem \ref{thm:quenched}}\label{sec:Proof:Quenched}

By Lemma~\ref{lem:Cor_calc}, the $qp$-correlations (\ref{eq:quenchcor}) corresponding to the time-evolution of any initially chosen
density matrix $\rho$ can be evaluated as:
\begin{eqnarray}
\tilde{\Gamma}_{\rho}(t) & = & \left\langle \begin{pmatrix} q \\ p \end{pmatrix} (q^T, p^T) \right\rangle_{\rho_t} - \left\langle \begin{pmatrix} q \\ p \end{pmatrix} \right\rangle_{\rho_t}
\left\langle  (q^T, p^T) \right\rangle_{\rho_t} \\ \nonumber
& = & V \begin{pmatrix} e^{-2i t \gamma} & 0 \\ 0 & e^{2it \gamma} \end{pmatrix} \Gamma_\rho^B \begin{pmatrix} e^{-2i t \gamma} & 0 \\ 0 & e^{2it \gamma} \end{pmatrix} V^T
\end{eqnarray}
where $V$ and $\Gamma_{\rho}^B$ are as in (\ref{def:V}) and (\ref{Cor_B}), respectively. For our arguments here, we prefer to re-express this in terms of
the time-zero $qp$-correlations, i.e., we write
\begin{equation} \label{3mats}
\tilde{\Gamma}_{\rho}(t) = V_t \Gamma_{\rho} V_t^T,
\end{equation}
where we have set
\begin{equation}
V_t = V \begin{pmatrix} e^{-2i t \gamma} & 0 \\ 0 & e^{2it \gamma} \end{pmatrix} V^{-1} = \begin{pmatrix} \cos(2 t h_\Lambda^{1/2}) & h_\Lambda^{-1/2} \sin(2 t h_\Lambda^{1/2}) \\ - h_\Lambda^{1/2} \sin(2 t h_\Lambda^{1/2}) & \cos(2 t h_\Lambda^{1/2}) \end{pmatrix}.
\end{equation}
The final equality is a direct calculation.

By (\ref{3mats}) one has, for any $x,y \in \Lambda$,
\begin{equation} \label{block_entries}
(\tilde{\Gamma}_{\rho}(t))_{xy}  = \sum_{z,z' \in \Lambda} (V_t)_{xz} (\Gamma_{\rho})_{zz'} (V_t^T)_{z'y}
\end{equation}
with $2\times 2$-matrices $(\tilde{\Gamma}_{\rho}(t))_{xy}$, $(V_t)_{xz}$,  $(\Gamma_{\rho})_{zz'}$ and $(V_t^T)_{z'y}$ defined according to (\ref{eq:block}).

Similar to the arguments in Section~\ref{sec:Proof:cluster}, our basic assumption (\ref{eq:eigenCorr}),  guarantees the existence of $\tilde{C}< \infty$, depending on $d$, $\lambda$, $s$ and $k_{max}$, such that
\begin{equation} \label{base_est}
\mathbb{E} \left( \sup_{t \in \mathbb{R}} \|(V_t)_{xy}\|^s \right) \leq \tilde{C} e^{- \eta|x-y|}
\end{equation}
for every rectangular box $\Lambda$ and all $x,y \in \Lambda$. It is clear that the same bound also holds for $V^T_t$.

For the product state $\rho= \otimes_{\ell=1}^M\rho_\ell$ in (\ref{eq:product}), the $qp$-correlation matrix is the direct sum of the correlation matrices of the factors $\rho_{\ell}$. More precisely, for $x,y \in \Lambda$,
\begin{equation}
(\Gamma_{\rho})_{xy} = \left\{  \begin{array}{ll}
   (\Gamma_{\rho_\ell})_{xy} & \mbox{if $x,y \in \Lambda_\ell$ for some $\ell$,}  \\
    0 & \hbox{otherwise.}
  \end{array}
\right.
\end{equation}
Thus, by condition (\ref{AssumLocal})
\begin{equation}\label{eq:Gammapq-2-Term}
\mathbb{E}\left(\|(\Gamma_{\rho})_{xy}\|^s\right)\leq C' e^{-\eta'|x-y|}
\end{equation}
for all $x,y \in \Lambda$.

For all $x,y \in \Lambda$ we have that
\begin{eqnarray} \label{Holder3}
\mathbb{E} \left( \sup_{t \in \mathbb{R}} \|(\tilde{\Gamma}_{\rho}(t))_{xy}\|^{s/3} \right) & \leq &  \sum_{z,z'\in \Lambda}
\mathbb{E} \left( \sup_{t \in \mathbb{R}} \|(V_t)_{xz}\|^s \right)^{1/3} \times \nonumber \\
& \mbox{ } & \quad \times \mathbb{E} \left( \|(\Gamma_{\rho})_{z z'}\|^s \right)^{1/3} \mathbb{E}
\left(\sup_{t \in \mathbb{R}}\|(V_t^T)_{z' y}\|^s \right)^{1/3},
\end{eqnarray}
where we have used (\ref{block_entries}) and H\"older's inequality. Thus (\ref{base_est}) and (\ref{eq:Gammapq-2-Term}) yield
\begin{eqnarray} \label{postHold}
\mathbb{E} \left( \sup_{t \in \mathbb{R}} \|(\tilde{\Gamma}_{\rho}(t))_{xy}\|^{s/3} \right) & \leq &  \tilde{C}^{2/3} C'^{\frac{1}{3}} \sum_{z,z' \in \Lambda}
e^{- \eta |x-z| /3} e^{- \eta' |z-z'|/3} e^{- \eta |z'-y|/3} \nonumber \\
& \leq & C'' e^{- \tilde{\eta}|x-y|}.
\end{eqnarray}
Here one may take $\tilde{\eta}=\frac{1}{6}\min \{\eta,\eta'\}$ and
\begin{equation} \label{eq:explconst}
C''= \tilde{C}^{2/3} (C')^{1/3} \left(\frac{2}{1-e^{-\tilde{\eta}}}\right)^{2d} .
\end{equation}

\section{Applications of Theorem~\ref{thm:quenched}} \label{sec:applications}

(i) As a first application of Theorem~\ref{thm:quenched} we consider the case where the factors in the product state (\ref{eq:product}) are thermal states of the subsystems. Assumption (\ref{AssumLocal}) in Theorem~\ref{thm:quenched} is then a consequence of Theorem~6.1 in \cite{NSS} on the position-momentum correlations of thermal states of oscillator systems. We start by stating an improved version of this result, which makes the temperature dependence of the bound explicit, a fact of some interest by itself which was not addressed in \cite{NSS}. In Appendix~\ref{sec:tempdep} we sketch the modifications of the argument in \cite{NSS} needed to get this improvement.

While more general systems are considered in \cite{NSS}, we will continue to focus on the model (\ref{eq:Ham}).
Here we only require the general assumption (\ref{eq:distr}) on the distribution of the $k_x$, and, in particular, we do not
require to be in a fully localized regime as needed for (\ref{eq:eigenCorr}).

\begin{thm}\label{thm:thermal:NSS}
For a rectangular box $\Lambda \subset \Z^d$ and $\beta\in(0,\infty)$, let $H_\Lambda$ be given by (\ref{eq:Ham}), $\rho_{\beta} = e^{-\beta H_{\Lambda}}/ \Tr [e^{-\beta H_{\Lambda}}]$ its thermal states, and $\Gamma_{\rho_\beta} = \Gamma_{\rho_\beta}(0,0)$ their static position-momentum correlation matrices.

There exist $C<\infty$ and $\mu>0$, dependent on $d$, $\lambda$ and the distribution of the random variables $k_x$, but independent of $\Lambda$ and $\beta$, such that
\begin{equation}\label{eq:NSS}
\mathbb{E}\left(\|(\Gamma_{\rho_{\beta}})_{xy}\|^{\frac{1}{2}}\right)\leq C \max\left\{1,\frac{1}{\beta}\right\}\ e^{-\mu|x-y|}
\end{equation}
for all $x,y \in \Lambda$.
\end{thm}

In Section~\ref{sec:tempdep} below we will briefly discuss how the $\beta$ dependence in (\ref{eq:NSS}) can be extracted from the bounds provided in \cite{NSS}.

Let us consider the quantum quench with respect to the decomposition $\Lambda =\biguplus_{\ell=1}^M\Lambda_\ell$, and assume that the local states are the thermal states of $H_{\Lambda_\ell}$ with inverse temperatures $\beta_\ell$, $\ell=1,\ldots,M$, i.e.,
\begin{equation}\label{eq:case1}
\rho_{\ell,\beta_\ell}=\frac{e^{-\beta_\ell H_{\Lambda_\ell}}}{\Tr[e^{-\beta_\ell H_{\Lambda_\ell}}]}.
\end{equation}
Then condition (\ref{AssumLocal}) is satisfied by Theorem 6.1 when applied to each of the local Hamiltonians.
In this case, and for the remainder of this section, we will further assume that (\ref{eq:eigenCorr}) holds with $s = 1/2$.
As is discussed in Section~\ref{sec:model}, this will be the case for the model we are considering when either $d=1$ or
$d \geq 1$ and the disorder is large. Now let
\begin{equation}
\rho_{\beta_1,\ldots,\beta_M}:=\bigotimes_{\ell=1}^M  \rho_{\ell,\beta_\ell}.
\end{equation}
Theorem \ref{thm:quenched} implies that
\begin{equation}\label{eq:dyn:thermal}
\mathbb{E}\left(\sup_t\|(\tilde{\Gamma}_{\rho_{\beta_1,\ldots,
\beta_M}}(t))_{xy}\|^\frac{1}{6}
\right)\leq C' \max \left\{ 1, \beta^{-1/3} \right\}  e^{-\tilde{\eta}|x-y|}
\end{equation}
for all $x,y\in\Lambda$. Here $\beta = \min_{\ell} \beta_{\ell}$ and $\tilde{\eta}=\frac{1}{6}\min\{\eta,\mu\}$ where $\eta$ and $\mu$ are as in (\ref{eq:eigenCorr}) and (\ref{eq:NSS}), respectively, and $C'$ is independent of $\Lambda$, $\beta$ and the number of subsystems $M$.

(ii) Next we discuss the case where the initial state is a product of eigenstates of the subsystems.
Fix a nonnegative integer $N < \infty$, and let $\alpha_{\ell}\in \mathbb{N}_0^{|\Lambda_\ell|}$ with $\|\alpha_{\ell}\|_\infty\leq N$ for all $\ell=1,\ldots,M$. Consider any family  $\rho_{\alpha_{\ell}} = |\psi_{\alpha_\ell} \rangle \langle \psi_{\alpha_\ell}|$, $\ell=1,\ldots,M$, of eigenstates of $H_{\Lambda_\ell}$ corresponding to the excitation vectors $\alpha_{\ell}$, with $\psi_{\alpha_\ell}$ given by (\ref{def:psi}) when used for the subsystem $\Lambda_{\ell}$ . Theorem \ref{thm:Corr:pq} implies that condition (\ref{AssumLocal}) is  satisfied for all $\ell=1,\ldots,M$, in particular, there exist constants $C'>0$  and $\eta<\infty$  such that
\begin{equation}
\mathbb{E}\left(\|(\Gamma_{\rho_{\alpha_\ell}})_{xy}\|^\frac{1}{2}\right)\leq C'(1+\|\alpha_\ell\|_\infty)^\frac{3}{2} e^{-\eta|x-y|},
\end{equation}
for all $\ell$ and all $x,y\in\Lambda_\ell$. Here $\eta$ is as in (\ref{eq:eigenCorr}), which we have again taken to hold with $s = 1/2$,
and $C'$ is independent of $\Lambda$, $N$, and of $M$.
With $\rho_{\alpha}$ the corresponding product, i.e. $\rho_{\alpha}=\otimes_{\ell=1}^M\rho_{\alpha_\ell}$,
an application of Theorem \ref{thm:quenched} shows that
\begin{equation}\label{eq:dyn:eigen}
\mathbb{E}\left(\sup_t\|(\tilde{\Gamma}_{\rho_{\alpha}}(t))_{xy}\|^\frac{1}{6}\right)\leq \tilde{C} (1+N)^{\frac{1}{2}} e^{-\frac{\eta}{6}|x-y|}
\end{equation}
for all $x,y\in\Lambda$. Here again $\tilde{C}<\infty$ is independent of $\Lambda$, the number $M$ of decompositions, and of the highest excitation $N$,
and moreover, $\eta$ is as above.

(iii)  One can combine the cases (i) and (ii) and consider a product state $\rho$ as in (\ref{eq:product}) where each local state is either a thermal state or an eigenstate of the Hamiltonian $H_{\Lambda_\ell}$. The arguments in cases (i) and( ii) above provide a proof of the following result, which summarizes all the cases considered so far.
\begin{cor}\label{Cor:Corr:Mix}
Fix $\beta>0$ and $N<\infty$. Let $\rho=\otimes_{\ell=1}^M\rho_\ell$ where each of the local states $\rho_\ell$ is either a thermal state of $H_{\Lambda_\ell}$ with inverse temperature $\beta_\ell\in[\beta,\infty)$, or an eigenstate associated with an excitation vector $\alpha_\ell$ such that $\|\alpha_\ell\|_\infty\le N$. If $\tilde{\eta}=\frac{1}{6}\min\{\eta,\mu\}$, where $\eta$ is as in (\ref{eq:eigenCorr}) with $s=1/2$ and $\mu$ is as in (\ref{eq:NSS}), then there exists $C<\infty$ such that
\begin{equation}\label{eq:dyn:mix}
\mathbb{E}\left(\sup_t\|(\tilde{\Gamma}_{\rho}(t))_{xy}\|^\frac{1}{6}
\right)\leq C \max\left\{(1+N)^\frac{3}{2}, \frac{1}{\beta} \right\}^{\frac{1}{3}} e^{-\tilde{\eta}|x-y|}
\end{equation}
for all $x,y\in\Lambda$. Here $C$ is independent of $\Lambda$, $N$, $M$ and of $\beta$.
\end{cor}

(iv) In the extreme case where each subsystem consists of only one site, i.e., $M=|\Lambda|$, the initial Hamiltonian $H_{0,\Lambda}$ is a system of non-interacting harmonic oscillators over the $d$ dimensional lattice $\Lambda$,
\begin{equation}
H_{0,\Lambda}=\sum_{x\in\Lambda} H_{\{x\}}\otimes \idty_{\Lambda\setminus \{x\}},
\end{equation}
where $H_{\{x\}}$ is the one dimensional harmonic oscillator
\begin{equation}\label{eq:Hx}
H_{\{x\}}=p_x^2+\frac{k_x}{2}q_x^2.
\end{equation}
The eigenstates of $H_{\{x\}}$ are known to be the Hermite functions
\begin{equation}\label{eq:phi-nu}
\phi_{n_x}(q_x)=\frac{1}{\sqrt{2^{n_x} n_x!}}\left(\frac{\sqrt{2k_x}}{2\pi}\right)^{\frac{1}{4}}
H_{n_x}(\sqrt[4]{\frac{k_x}{2}}q_x)
e^{-\frac{\sqrt{2k_x}}{4}q_x^2},
\end{equation}
where $n_x\in\mathbb{N}_0$ the excitation number at vertex $x\in\Lambda$, $H_{n_x}(\cdot)$  is the Hermite polynomial of degree $n_x$. In this special case, the following corollary improves on the bound in (\ref{eq:dyn:eigen}) for the correlations of the dynamics of the product state
\begin{equation}\label{eq:Corr:rho:extreme}
\rho=\bigotimes_{x\in\Lambda}\rho_{n_x}, \text{ where }
\rho_{n_x}=|\phi_{n_x}
\rangle\langle\phi_{n_x}|.
\end{equation}

\begin{cor}\label{cor:extreme:M=n}
Let $\rho$ be as in (\ref{eq:Corr:rho:extreme}) and let $N= \max_x n_x$. If $\eta$ is as in (\ref{eq:eigenCorr}) with $s=1/2$, then there exists $C<\infty$, independent of $\Lambda$ and $N$, such that
\begin{equation} \label{triv_dec}
\mathbb{E}\left(\sup_t\|(\tilde{\Gamma}_{\rho}(t))_{xy}\|^{\frac{1}{6}}\right)\leq C (1+2N)^{\frac{1}{6}}e^{-\frac{\eta}{6}|x-y|}
\end{equation}
for all $x,y\in\Lambda$.
\end{cor}

\begin{proof}
In the current case the local Hamiltonians  $h_{\{x\}}$ reduce to the single numbers $k_x/2$. This means that in the case of an $n_x$ excitation at site $x$,  the correlation matrix from (\ref{gen_cor_alpha}) reduces to the $2\times 2$-matrix
\begin{equation}
\Gamma_{n_x}(0)=\begin{pmatrix}
                  \frac{1}{\sqrt{2}}k_x^{-\frac{1}{2}}(1+2n_x) &
                  \frac{1}{2}i \\
                  -\frac{1}{2}i &
                  \frac{1}{2\sqrt{2}}k_x^{\frac{1}{2}}(1+2n_x) \\
                \end{pmatrix}.
\end{equation}
Since $k_x$ is a random variable with a bounded density $\nu$ and supported on the compact set $[0,k_{\max}]$, one gets
\begin{equation}
\mathbb{E}(\max\{k_x^{\frac{1}{2}},k_x^{-\frac{1}{2}}\})\leq 2\|\nu\|_{\infty}\max\left\{(k_{\max})^\frac{1}{2},(k_{\max})^{\frac{3}{2}}/3\right\}.
\end{equation}
 Hence, there exists $C'<\infty$ such that
\begin{equation}\label{eq:Hx:Gamma0}
\mathbb{E}\left(\|\Gamma_{n_x}(0)\|\right)\leq C' (1+2n_x).
\end{equation}

Given (\ref{eq:Hx:Gamma0}), the full correlation matrix for the product state $\rho = \otimes_x \rho_{n_x}$ satisfies
\begin{equation}\label{eq:Cor:Cqp:2-Term}
\mathbb{E}\left(\|(\Gamma_{\rho})_{xy}\|\right)\leq C'(1+2N)\delta_{x,y}.
\end{equation}
Arguing as in (\ref{Holder3}), see also (\ref{postHold}), we conclude that
\begin{equation}
\mathbb{E}\left(\sup_{t\in\mathbb{R}}
\|(\tilde{\Gamma}_{\rho}(t))_{xy}\|^{\frac{1}{6}}\right) \leq \tilde{C}^\frac{2}{3} \sum_{z,z' \in \Lambda} e^{- \eta|x-z|/3} \mathbb{E}\left(\|(\Gamma_{\rho})_{zz'}\|^\frac{1}{2} \right)^\frac{1}{3} e^{- \eta |z'-y|/3}
\end{equation}
where $\tilde{C}$ is as in (\ref{base_est}) with $s=1/2$. Since Holder and (\ref{eq:Cor:Cqp:2-Term}) imply that
\begin{equation}
\mathbb{E}\left(\|(\Gamma_{\rho})_{zz'}\|^\frac{1}{2} \right) \leq \mathbb{E}\left(\|(\Gamma_{\rho})_{zz'}\| \right)^\frac{1}{2} \leq \sqrt{C'(1+2N)} \delta_{z,z'}
\end{equation}
the claim in (\ref{triv_dec}) now follows as in the end of the proof of Theorem~\ref{thm:quenched}.
\end{proof}

\appendix

\section{Proof of Theorem~\ref{thm:thermal:NSS}} \label{sec:tempdep}

In the following we use and refine several results from \cite{NSS} to prove Theorem~\ref{thm:thermal:NSS}. We start by noting that these results are only formulated for cubes in \cite{NSS}, but that they extend to the rectangular boxes considered here.

By Lemma~5.4 of \cite{NSS},
\begin{equation}
(\Gamma_{\rho_\beta})_{xy} = \frac{1}{2} \begin{pmatrix} \langle \delta_x, h_{\Lambda}^{-1/2} \varphi(h_\Lambda) \delta_y \rangle & i \delta_{xy} \\ -i \delta_{xy} & \langle \delta_x, h_\Lambda^{1/2} \varphi(h_\Lambda) \delta_y \rangle \end{pmatrix},
\end{equation}
where we have set $\varphi(t) = \coth(\beta t^{1/2})$. Thus for (\ref{eq:NSS}) it suffices to show that
\begin{equation} \label{eq:toshow}
\E \left( |\langle \delta_x, h_{\Lambda}^{\pm 1/2} \varphi(h_\Lambda) \delta_y \rangle|^{1/2} \right) \le C \max\left\{1,\frac{1}{\beta}\right\}\ e^{-\mu|x-y|}.
\end{equation}
 Expanding $\idty = \sum_z |\delta_z \rangle \langle \delta_z|$ and an application of H\"older's inequality show that the left hand side of (\ref{eq:toshow}) can be bounded by
\begin{equation} \label{eq:ddd}
\sum_z \left( \E (|\langle \delta_x, h_\Lambda^{\pm 1/2} \delta_z \rangle|) \right)^{1/2} \left( \E (| \langle \delta_z, \varphi(h_\Lambda) \delta_y \rangle |) \right)^{1/2}.
\end{equation}

The two factors in the sum can both be bounded using Proposition~A.3(c) of \cite{NSS} and the method of its proof, respectively. For the first factor we can cite Proposition~A.3(c) directly to conclude the existence of $C_1<\infty$ and $\mu_1>0$ such that
\begin{equation} \label{eq:aaa}
\E (|\langle \delta_x, h_\Lambda^{\pm 1/2} \delta_z \rangle|) \le C_1 e^{-\mu_1 |x-z|}
\end{equation}
for all $x$ and $z$.

To understand the $\beta$-dependence of the second factor, we need to analyze the proof of Proposition~A.3(c) of \cite{NSS}. It requires splitting low and high energies of $h_\Lambda$. At low energies, we can use localization of $h_\Lambda$: Our assumptions yield that there exists $E_0>0$ such that the Green function of $h_{\Lambda}$ has localized $s$-fractional moments in $[0,E_0]$ for all $s\in (0,1)$, e.g.\ \cite{AizenmanWarzel}. By Proposition A.3(b) of \cite{NSS} this implies the existence of $C_2<\infty$ and $\mu_2>0$ such that
\begin{equation} \label{eq:locbound}
\E \left( \sup_{|g|\le 1} |\langle \delta_z, h_{\Lambda}^{-1/2} g(h_\Lambda) \chi_{[0,E_0]}(h_\Lambda) \delta_y \rangle| \right) \le C_2 e^{-\mu_2 |z-y|}.
\end{equation}
Using the elementary bound
\begin{equation}\label{eq:coth01}
|\varphi(t)|\leq \frac{\beta\sqrt{E_0}+1}{\beta}t^{-\frac{1}{2}}
\end{equation}
for all $t\in[0,E_0]$, (\ref{eq:locbound}) gives
\begin{equation} \label{eq:bbb}
\mathbb{E}\left(|\langle \delta_z,\varphi(h_\Lambda)\chi_{[0,E_0]}(h_\Lambda)\delta_y\rangle |\right)\leq C_2 \frac{\beta\sqrt{E_0}+1}{\beta} e^{-\mu_2 |z-y|}.
\end{equation}

We will further prove that there are $C_3<\infty$ and $\mu_3>0$ such that
\begin{eqnarray} \label{eq:ccc}
\mathbb{E}\left(|\langle\delta_z,\varphi(h_\Lambda)\chi_{(E_0,\infty)}(h_\Lambda)\delta_y\rangle|\right) & \leq & C_3 \coth^2(\beta\sqrt{E_0})\ e^{-\mu_3 |z-y|} \\ & \le & C_3 \left( 1 + \frac{1}{\beta \sqrt{E_0}} \right)^2 \ e^{-\mu_3 |z-y|} \nonumber
\end{eqnarray}
for all $z$ and $y$.

Inserting all of (\ref{eq:aaa}), (\ref{eq:bbb}) and (\ref{eq:ccc}) into (\ref{eq:ddd}), ultimately gives the bound (\ref{eq:toshow}).

We still owe the proof of the first claim in (\ref{eq:ccc}). This is done by an analysis of the proof of Proposition~A.3(c) in \cite{NSS}. This proof, see (A.15) in \cite{NSS}, uses that
\begin{equation} \label{eq:contint}
|\langle\delta_z,\varphi(h_\Lambda)\chi_{(E_0,\infty)}(h_\Lambda)\delta_y\rangle| \le C' \int_{\Gamma} |\langle \delta_z, (h_{\Lambda}-\zeta)^{-1} \chi_{(E_0,M]}(h_{\Lambda}) \delta_y \rangle|\,|d\zeta|,
\end{equation}
where $\Gamma$ is the rectangular contour with vertices $E_0\pm i$ and $(M+1)\pm i$, $M$ the a-priori upper bound for $\sigma(h_{\Lambda})$ from (\ref{eq:spec}), and $C'= \max \{|\varphi(\zeta)|: \zeta\in \Gamma\} /(2\pi)$.

Using the elementary bound $|\coth(\zeta)|\leq \coth^2(\RE \zeta)$ for $\zeta \in\mathbb{C}\setminus\{0\}$ one has
\begin{equation}\label{eq:varphibound}
|\varphi(\zeta)| \leq \coth^2(\beta\RE \zeta^\frac{1}{2})=\coth^2\left(\beta\sqrt{\frac{\RE \zeta+|\zeta|}{2}}\right).
\end{equation}
This means that
\begin{equation}\label{eq:CprimeBound}
C' \leq \frac{1}{2\pi}\coth^2\left(\beta\min_{\zeta\in\Gamma}\sqrt{\frac{\RE \zeta+|\zeta|}{2}}\right)
= \frac{1}{2\pi}\coth^2\left(\beta\sqrt{E_0}\right).
\end{equation}

The argument in \cite{NSS} shows that the ($\beta$-independent) integral in (\ref{eq:contint}) is bounded by $C'' e^{-\mu_3 |z-y|}$ for some $C''<\infty$, $\mu_3>0$. Combined with (\ref{eq:CprimeBound}) this yields (\ref{eq:ccc}).

\end{document}